\documentclass[letterpaper, 10 pt, conference]{ieeeconf}  

\IEEEoverridecommandlockouts                              

\usepackage{amsmath} 
\usepackage{amssymb}  
\usepackage{graphicx}
\usepackage{xcolor}
\usepackage{algorithmic}
\usepackage{algorithm}
\usepackage{accents}
\usepackage{enumerate}
\usepackage{nicefrac}
\usepackage{afterpage}
\newcommand{\ubar}[1]
{\underaccent{\bar}{#1}}
\newcommand{\R}{\mathcal{R}}
\newcommand{\B}{\mathcal{B}}
\newcommand{\nbd}[2]{\mathcal{N}^\beta_{#1}(#2)}
\title{\Large \bf
ABRA: An algorithm which cannot converge to low-quality Nash equilibria
}

\author{Vartika Singh and Philip N. Brown
\thanks{This material is based upon work supported by the Air Force Office of Scientific Research under award number FA9550-23-1-0171.}
\thanks{The authors are with the University of Colorado at Colorado Springs, CO 80918, USA.
	{\tt\small \{vsingh, pbrown2\}@uccs.edu}}%
}

\newtheorem{theorem}{Theorem}
\newtheorem{corollary}{Corollary}
\newtheorem{lemma}{Lemma}
\newtheorem{conjecture}{Conjecture}
\newtheorem{proposition}{Proposition}
\newcommand{\A}{\mathcal{A}
}
\newcommand{\hide}[1]{}
\renewcommand{\a}{{\bf a}}
 \newcommand{\eop}{{\hfill $\blacksquare$}}

\renewcommand{\P}[2]{\mathcal{P}_{#2}(#1)}
\begin{document}

\maketitle
\thispagestyle{empty}
\pagestyle{empty}

\begin{abstract}

We consider a game theoretic approach to solve multi-agent coordination problems with submodular objectives. It is known for such problems that the Nash equilibria for the corresponding game are always within 50\% of the optimal. A recent work further shows that the equilibria which achieve this worst-case bound are not stable. Leveraging this, we design an Approximate Best Response Algorithm (ABRA) governed by a noise parameter and a rationality parameter. The noise allows ABRA to escape the bad equilibria and the rationality parameter balances any degradation in the objective function caused by the noise. 
We show for any two-player game that if ABRA converges to a Nash equilibrium, its system objective value is strictly more than 50\% of optimal plus a term controlled by the noise parameter.
Otherwise, ABRA converges to some recurrent class: if a recurrent class contains any action profile yielding system objective less than 50\% of the optimal, the class must also contain  either the optimal action profile or an action profile yielding system objective strictly more than 50\% of the optimal by the same amount in addition to a factor controlled by noise parameter. The time that ABRA spends in such action profiles can be controlled using the rationality parameter.
Using numerical simulations, we show that 
the minimum expected objective function is typically well above half of the optimal.

\end{abstract}

\section{INTRODUCTION}

Multi-agent coordination problems have a wide range of applications such as task assignment, resource sharing e.g., channel access control in wireless networks, coverage optimization, network routing etc. (see \cite{Qu}-\cite{Kordonis}). The goal is to maximize a system objective controlled by the actions of all the agents. A natural approach to solve such problems uses game theory (see \cite{marden}-\cite{martin}) -- a system planner endows the agents with individual utility functions and decision rules that depend only on their local information. The system planner can design the local utility functions in order to drive the equilibria of the game to a desired system optimal (e.g., \cite{marden_util_des}). 

One of the utility designs in this context is the marginal contribution, where the utility of any agent equals the gain in the system objective when the agent decides to participate. This choice of utility design leads to a potential game with the system objective as the potential function (see e.g. \cite{marden_potential}). It is well known that the potential function and thus the system objective is maximized at one of the pure Nash equilibria (NE) of the corresponding potential game. However, in case of multiple NE, the agents may arrive to an undesired equilibrium with inferior system objective value. Thus, one requires a metric to measure the performance of any game, for example, \textit{price of anarchy} (see \cite{PoA}) which is defined to be the ratio of system objective at the system optimal to that at the worst NE.

{When the system objective is sub-modular and the marginal contribution utility design is considered,  the resulting game becomes a \textit{valid utility game} (see \cite{vetta})}. For this special class of games, \cite{vetta} proves that the price of anarchy (PoA) is at least $\frac{1}{2}$. \hide{Authors in \cite{marden_util_des} develop a framework that characterises the PoA for any choice of utility design, and focus on optimizing the PoA by choosing appropriate utility functions. In contrast to \cite{marden_util_des}, we fix the utility design to the marginal contribution and focus on designing an algorithm such that system performance is improved.} Further, \cite{seaton} show that, a) any bad NE  with PoA close to $\frac{1}{2}$ is not stable and, b) any stable NE has a PoA significantly more than $\frac{1}{2}$ for the valid utility games. In other words, a stable NE cannot be very bad, and a bad NE cannot be stable. 
Thus  it may be possible to avoid an undesirable NE  by adding some noise in the system.

In this work, we  propose  Approximate Best Response Algorithm (ABRA) which allows the agents to choose actions which are approximately best responses, controlled by a noise parameter $\beta$ and rationality parameter $p$.   At any iterate,  the updating agent either chooses among the best response actions with probability governed by rationality parameter $p$ or chooses an action uniformly from the $\beta$-neighbourhood of the best responses, excluding the best response set.%
%

For all two-player valid utility games with marginal contribution utilities, we prove that if the sequence of action profiles generated using ABRA algorithm converges to a single action profile, then it must be an NE with the system objective value within $\frac{1}{2}$ of the optimal  plus a term increasing in $\beta$. This implies that if ABRA converges to an NE, that NE is guaranteed to exceed the worst-case PoA bound at least by an amount that can be controlled with the noise parameter $\beta$. Further, if the sequence converges to a recurrent class of action profiles that does not contain the optimal action profile (which we call a \textit{sub-optimal class}), we prove that {the sub-optimal class always contains an NE with system objective value strictly more than $\frac{1}{2}$ of the optimal plus a term increasing with noise parameter $\beta$. 
Otherwise, ABRA converges to a recurrent class that contains the optimal action, and is  an \textit{optimal recurrent class}. Interestingly, the lower bound on system objective for both the optimal recurrent class  and sub-optimal class degrades with increasing noise.  Thus, there is trade-off between improving the quality  of reachable NE, and degrading the lower bound on optimal or sub-optimal class by choosing a noise parameter $\beta$.} 

%
%
%
Using numerical simulations, we depict that for any noise parameter  $\beta$, an appropriate choice  of rationality parameter $p$ ensures that the expected value of the system objective  is well within $\frac{1}{2}$ of the optimal plus a term increasing with noise parameter $\beta$.  Interestingly, when the worst case lower bound in the optimal class is achievable, {the lower bound on system objective  of the NE in any sub-optimal class  is guaranteed to be significantly high} (see Figure \ref{fig:tradeoff}). In all, one can control the quality of optimal and sub-optimal classes using the noise parameter $\beta$ and rationality parameter~$p$.

\begin{figure}
    \centering
    \includegraphics[trim ={3cm 8cm 3cm 9cm}, clip, scale=0.5]{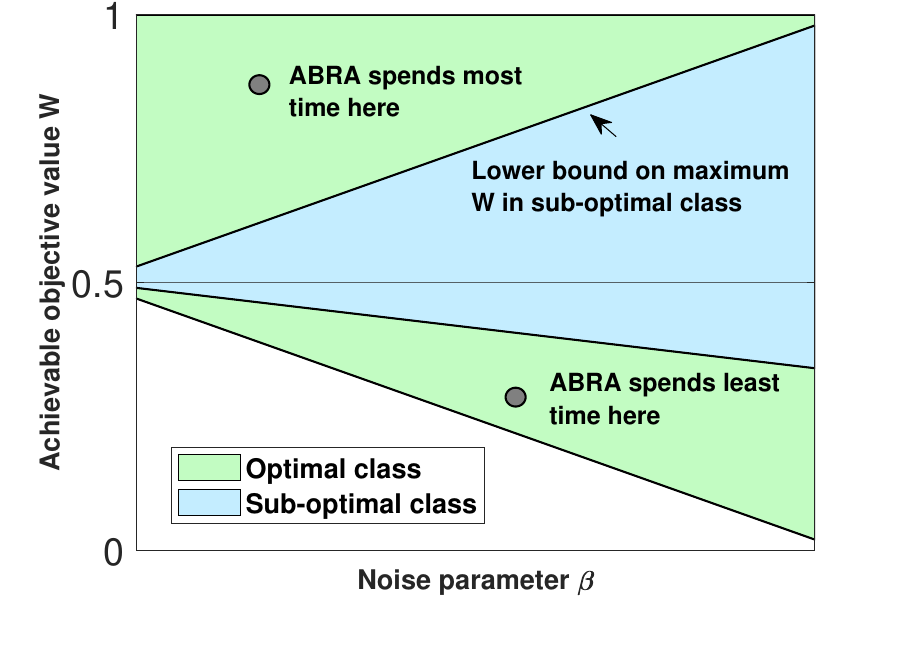}
    \caption{The values that the system objective function can take once the Approximate Best Response Algorithm (ABRA) reaches a recurrent class: (i) The green shaded region corresponds to achievable values of system objective in the worst case scenario if ABRA converges to a recurrent class that contains an optimal action; although the achievable action profiles may contain values less than $\frac{1}{2}$ of the optimal, an appropriate choice of rationality parameter $p$ can ensure that ABRA spends limited time in such bad states (see numerical experiments in Section~\ref{sec: numerical}). (ii) The blue shaded region corresponds to achievable values of system objective when ABRA converges to a sub-optimal recurrent class in the worst case scenario. There is a trade-off between improving the quality of sub-optimal classes and degrading the quality of optimal classes based on noise parameter~$\beta$.}
    \label{fig:tradeoff}
\end{figure}

\section{Model}

We consider a game theoretic approach for sub-modular maximization problems. A system designer provides two agents with local utility functions and decision rules that depend only on the information available to the agents. The action set for any agent $i$ is given by  $ \A_i:=\{a_i^1,\dots,a_i^{n_i}\}$ where $n_i<\infty$, and the system objective is specified by a function $W:\A_1\times\A_2 \to [0,1]$. We assume the function $W$  to be sub-modular, non-decreasing and normalised; that is, for any sets $a, a'$ and $b$ such that $a \subset a'$, we have (with $(a,b)$ representing $a \cup b$),
\begin{eqnarray}\label{eqn_sub_mod}
    W(a,b) - W(a) &\ge &  W(a', b) - W(a'), \nonumber\\
    W(a')\ge W(a) &\mbox{and}& W(\varnothing)=0.
\end{eqnarray}  

Any agent $i$ can observe the action of the other agent but is unaware of their action set $\A_{-i}$, thus cannot directly compute the optimal action profile.  The system designer aims to design the utility functions and decision rules such that the agent behaviour is driven towards a desirable outcome.  In this paper, we consider the local utility functions of the agents to be modelled as  the marginal contribution (see \cite{marden_potential}):
\begin{eqnarray}\label{eqn_util}
    U_i(a_i,a_{-i}) := W(a_i,a_{-i}) - W(a_{-i}),
\end{eqnarray}where $a_i$ represents the action of agent $i$ and $a_{-i}$ represents the  action  of  the other agent. It is clear from \eqref{eqn_util} that the local utility of any agent equals the gain in the system utility when that agent chooses to participate. One may anticipate that when agents act as per their self-interest, and try to improve their local utilities, they may also improve the system objective. 

 The marginal contribution utility structure \eqref{eqn_util} induces a potential game among the agents with the system objective function $W$ as the potential function (see \cite{marden_potential}). It is well known that one of the pure Nash equilibrium (NE) of a potential game is the maximizer of the potential function, and thus the maximizer of $W$. However, the game may have multiple NE, and the NE other than the optimal can be inferior and degrade the system objective value.
 It is shown in \cite{vetta} that any NE is guaranteed to be within $\frac{1}{2}$ of the optimal for such games. Further, \cite{seaton} shows that the `bad' NE (the ones with inferior system objective) are not stable -- at an unstable NE, none of the agents lose much by switching to the optimal action, i.e., the optimal action is an approximate best response for any agent. Inspired by this, we design an algorithm that leverages the fragility of the `bad' NE in order to escape them by adding a noise parameter $\beta$. We describe the algorithm in the next section.

\section{Approximate Best Response Algorithm}\label{sec_algo}

From \cite{seaton}, since the `bad' NE are unstable, one may be able to   escape them by adding some noise while considering the best responses of the agents. We add a noise parameter $\beta \in [0,0.5)$ and allow the agents to choose the approximate best response from the $\beta$-neighbourhood of their best response\footnote{A noise parameter $\beta\ge 0.5$ does not add any value to the system and makes all  NE recurrent.}.  Basically, the agent chooses its action (in a probabilistic manner) from the set of actions which provide utility within $\beta$  of the best possible utility for the given action of other agent. 
A high noise parameter $\beta$ may degrade the performance of the system by allowing inferior actions in the $\beta$-neighbourhood, while a low noise parameter may not be sufficient to escape the bad NE.  Our approach allows us to explicitly characterize this trade-off. 


We first define a few notations. Let $\B_i(a_{-i})$ represent the best response set of agent $i$ when action of other agent is $a_{-i}$;  i.e., $\B_i(a_{-i}):= \arg\max\{U_i(a_i,a_{-i}): a_i \in \A_i\}$. Further, define $\nbd{i}{a_{-i}}$ to be the set of actions in the  $\beta$-neighbourhood of the best response  set for agent $i$ \textit{excluding} the best response actions (with $a^* \in \B_i(a_{-i})$):
\begin{equation*}
  \nbd{i}{a_{-i}} := \{a_i: |U_i(a^*, a_{-i})-U(a_i,a_{-i})| \le \beta\} \backslash \B_i(a_{-i}).
\end{equation*}%
Lastly, let $\a(t):=(a_1(t),a_{2}(t))$ denote the joint action profile of the agents at time $t$.

Both the agents start at some initial joint action profile $\a(0) = (a_1(0),a_2(0))$. At any $t \ge 1$, an agent (say $i$) is selected uniformly from the set of agents, who  observes the joint action profile $\a(t-1)$, and updates its own action to $a_i(t)$ as follows,
\begin{eqnarray*}
    a_i(t) =\left\{ \begin{array}{ll}
       a \in \B_i(a_{-i}(t-1))  &  \mbox{ with probability } p,\\
         a \in \nbd{i}{a_{-i}(t-1)}  &  \mbox{ with probability } 1-p.\\ 
    \end{array}\right.
\end{eqnarray*}In the above, $a$ is chosen from $\B_i(a_{-i}(t-1))$, or $\nbd{i}{a_{-i}(t-1)}$ in a uniform manner\footnote{Note that the usual best response dynamics algorithm is a special case of ABRA when $\beta=0$.}. The rationality parameter $p$ satisfies $p<1$, so that the actions from $ \nbd{i}{a_{-i}(t-1)} $ are chosen with positive probability. As time progresses, the sequence of joint action profiles $\{\a(t)\}_t$ evolves as a Markov chain (see e.g., \cite{Levin,Norris}) governed by the noise parameter $\beta$ and rationality parameter $p$.  This Markov chain need not be irreducible, and may have multiple recurrent classes further depending upon the structure of $W$.   We refer to any recurrent class that contains an optimizer of $W$ as an \textit{optimal class} (denoted by $\mathcal{R}_\beta^*$) and  any recurrent class that does not contain the optimizer as a \textit{sub-optimal class} (denoted by $\mathcal{R}_\beta$). The Markov chain may converge to an absorbing state, a sub-optimal class or an optimal class depending upon initial distribution. We now provide the qualitative analysis of the system objective over these classes generated under ABRA with noise parameter $\beta$ and rationality parameter $p$.

\section{Main results}\label{sec:main_results}
Without loss of generality, assume that $(a^1_1,a_2^1)$ is the optimal profile with $W(a^1_1,a_2^1)=1$ and define
\begin{equation}\label{eqn_x_bar}
   \ubar{x}:= \min\{ W(a_1^j,a_2^1),W(a^1_1,a_2^k): j \in \mathcal{A}_1, k\in \mathcal{A}_2 \}, 
\end{equation}
to be the minimum value of system objective function that can be achieved when any agent unilaterally deviates from the optimal.
\textit{For ease of explanation, say $\ubar{x}$ is achieved by  unilateral deviation by agent 1, i.e., $\ubar{x}=W(a_1^k,a_2^1)$ for some $a_1^k \in \mathcal{A}_1$}. The analysis follows in the exact similar manner if $\ubar{x}$ was achieved by the deviation of agent 2. 

\subsection{Absorbing states}
If the sequence of  action profiles $\{\a(t)\}$ generated using ABRA converges to an absorbing state, then we have the following result  with proof in Appendix:

\begin{theorem}\label{thm_NE}
     If $(a_1^*,a_2^*)$ is an absorbing state for Markov chain $\{(a_1(t),a_2(t))\}_t$ generated using ABRA with noise parameter $\beta$, then, i) $(a_1^*,a_2^*)$ is an NE, and  ii) the system objective $W(a_1^*,a_2^*) > \frac{1}{2}+\beta$. 
\end{theorem} 
\vspace{2mm}

The above result implies that the sequence of action profiles $\{(a_1(t),a_2(t))\}_t$ generated using ABRA converges only to the NE that have system objective value within $\frac{1}{2}+\beta$ of the optimal, recall the optimal value is 1. This provides a significant improvement over the known bound \cite{vetta} of any NE being within $\frac{1}{2}$ of the optimal. Also, note that the above result is true for any choice of rationality parameter $p<1$. The above result also implies that any NE with system objective value as $\frac{1}{2}$ of the optimal is never achieved even when noise parameter $\beta=0$. This reaffirms the findings of \cite{seaton} regarding the fragility of such bad NEs.

\subsection{Sub-optimal classes}
Now, we consider the sub-optimal classes -- recall that any sub-optimal class is a recurrent class of the action profiles generated using ABRA with noise and rationality parameter $(\beta,p)$, that does not contain the optimal action profile. We have the following result with proof in Appendix.

{\begin{theorem}\label{thm_abs_state}
If $\mathcal{R}_\beta$ is a sub-optimal class under ABRA with noise parameter $\beta$, then
\begin{eqnarray}\label{eqn_bdd_min_rec}
    \min_{(a_1,a_2) \in \R_\beta} W(a) &>&\frac{1}{2}-\frac{\beta}{2}.
\end{eqnarray}Further, if for some $0 \le \delta< \frac{\beta}{2}$, $\min_{(a_1,a_2) \in \R_\beta} W(a) = \frac{1}{2}-\delta$, then,
  \begin{eqnarray}\label{eqn_new_max_two_p}
      \max_{(a_1,a_2)  \in \R_\beta} W(a) &>& \frac{1}{2}+\delta + \beta.
  \end{eqnarray}
\end{theorem}}
\vspace{2mm}

From the above theorem, it is guaranteed that {the  sub-optimal class always contains a Nash equilibrium that yields system objective value  more than $\frac{1}{2}$ of the optimal  by a factor $\beta$. Moreover, when the sub-optimal class contains an action profile with system objective less than $\frac{1}{2}$ of the optimal, then a) the quality of the Nash equilibrium is further improved by the same amount, and b) the system objective of any action profile in  sub-optimal class can not be worse than $\frac{1}{2}$ of the optimal by a factor more than $\frac{\beta}{2}$.  Thus, when $\{\a(t)\}_t$ converges to any sub-optimal class, the expected system objective can be  strictly more than $\frac{1}{2}$ of the optimal by an appropriate choice of rationality and noise parameter.}

\subsection{Optimal classes}
{We now consider any optimal class $\mathcal{R}^*_\beta$ -- the recurrent class of action profiles generated using ABRA that contains the optimal action profile $(a^1_1,a^1_2)$. We derive a lower bound on the system objective value over the optimal class; this bound is not arbitrarily bad. We have the following result with proof in the Appendix.}

\begin{lemma}\label{thm_rec_state}
If $\mathcal{R}_\beta^*$ is an optimal class of action profiles under ABRA with noise parameter $\beta$, then the system objective function satisfies the following:
\begin{equation}
    W(a_1,a_2) \ge \max\{\max\{1-\ubar{x}, \ubar{x}\}-2\beta, \ubar{x} - \beta\}
\end{equation} for all $(a_1,a_2) \in \mathcal{R}^*_\beta$, with $\ubar{x}$ as in \eqref{eqn_x_bar}.  
\end{lemma}
\vspace{2mm}

The above result provides the worst possible value that the system objective function can take once the sequence $\{\a(t)\}_t$ enters the optimal class. The lemma shows that this lower bound depends upon the particular game in consideration through $\ubar{x}$. Now, we provide a lower bound on system objective irrespective of the game in consideration for any  noise parameter $\beta$ (proof in Appendix).

\begin{theorem}\label{cor_rec_state}
If $\mathcal{R}_\beta^*$ is an optimal class of action profiles under ABRA with noise parameter $\beta$, then the system objective function satisfies the following:
\begin{equation}\label{eqn_wrst_case_optimal}
 W(a_1,a_2) \ge  \frac{1}{2}-\frac{3\beta}{2} 
\end{equation}
     for all $(a_1,a_2) \in \mathcal{R}^*_\beta$. This lower bound is achievable only when $\ubar{x}$ in \eqref{eqn_x_bar} equals $\frac{1}{2}-\frac{\beta}{2}$. 
\end{theorem}
\vspace{2mm}

The above theorem shows that the minimum value that the system objective can take if $\{\a(t)\}_t$ converges to $\mathcal{R}_\beta^*$ can be as bad as in \eqref{eqn_wrst_case_optimal}. This may seem like a negative result but we discuss  in  section \ref{sec:num_rationality_parameter} that an appropriate rationality parameter limits the time ABRA spends in the action profiles with worst system objective value. Before that, we describe the trade-off between improving the sub-optimal classes and degrading the optimal class in the following.

\subsection{Trade-off between optimal and sub-optimal classes}
Whenever the lower bound in \eqref{eqn_wrst_case_optimal}  of Theorem \ref{cor_rec_state} is achievable, surprisingly, the quality of the sub-optimal classes significantly improves as shown in the following result with proof in the Appendix:

{\begin{proposition}\label{prop_low_bound_subopt}
    If the  bound on $W$ in  \eqref{eqn_wrst_case_optimal} under ABRA with noise parameter $\beta$ is achievable, then for any sub-optimal class $\mathcal{R}_\beta$,   we have   
    \begin{eqnarray}\label{eqn_low_bdd_sub_opt}
        \max_{(a_1,a_2) \in \mathcal{R}_\beta} W(a_1,a_2) &>& \frac{1}{2} + \frac{3\beta}{2} .
    \end{eqnarray}
\end{proposition}}
\vspace{2mm}

In all, Theorem \ref{cor_rec_state} and Proposition \ref{prop_low_bound_subopt} provide a trade-off between degrading the quality of optimal class and improving the quality of sub-optimal classes based on the choice of noise parameter $\beta$. A higher $\beta$ provides better lower bounds over the {Nash equilibria} in the sub-optimal classes \eqref{eqn_low_bdd_sub_opt} but degrades the quality of action profiles in  the optimal class \eqref{eqn_wrst_case_optimal}; this trade-off is also depicted in Figure \ref{fig:tradeoff}. 

\hide{\subsection{Expected system objective under ABRA}\label{sub_sec_expected_vals}
{Theorem \ref{thm_NE} and Theorem \ref{thm_abs_state} show that the minimum value that system objective $W$ can take over absorbing states is strictly more than $\frac{1}{2}$}. However, when the sequence $\{\a(t)\}_t$ converges to {\color{blue} an optimal class or sub-optimal class}, the system objective  can take values less than $\frac{1}{2}$ as shown in  {\color{blue}Theorem \ref{thm_abs_state} and Theorem \ref{cor_rec_state}}. It is important to note that the actual performance of ABRA depends upon the expected time that the $\{\a(t)\}_t$ spends in such minimizing action profiles. This expected time is governed by the rationality parameter $p$, and in Section~\ref{sec: numerical} we  numerically show that  ABRA spends most of the time choosing optimal actions 
 once rationality parameter $p$ is sufficiently high. {\color{blue}We discuss the case of optimal class in the following.} Let $\pi^*_p$ represent the stationary distribution concentrated on an optimal class for ABRA with noise parameter $\beta$ and rationality parameter $p$. Let $q^*_p$ represent the probability of choosing an optimal action profile under stationary distribution $\pi^*$, for example, $q^*_p = \pi^*_p(a^1_1,a^1_2)$. We show numerically that $q^*_p $ increases with $p$. Based on this observation, we make the following conjecture whose proof is the left for the future work.
 \begin{conjecture}\label{conjec}
     The probability of choosing an optimal action $q^*_p$ under stationary distribution $\pi^*_p$ concentrated on an optimal class increases with rationality parameter $p$ for any fixed noise parameter $\beta$. 
 \end{conjecture}
 \vspace{2mm}
 
 In the following, we show that the expected value of system objective under stationary distribution $\pi^*_p$ is guaranteed to be strictly more than $\frac{1}{2}$ plus a term controlled by noise parameter once $q_p^*$ is sufficiently large (proof in Appendix).

\begin{proposition}\label{prop_lower_bound_prob}
 For any fixed $\beta$, if rationality parameter $p$ is chosen such that $q_p^* > \frac{5\beta}{1+3\beta}$, then expected system objective value under stationary distribution concentrated on the optimal class $\pi^*_p$,
 \begin{eqnarray}
  E_{\pi^*_p}[W] > \frac{1}{2}+ \beta.  
 \end{eqnarray} 
\end{proposition}
\vspace{2mm}

In all, for any fixed $\beta$, one can anticipate that  increasing the rationality  parameter $p$ results in higher values of $q^*_p$. Then by Proposition \ref{prop_lower_bound_prob}, the expected value of system objective over optimal class is more than $\frac{1}{2}+\beta$. It is easy to verify that this also holds true for expected value of system objective over sub-optimal classes and absorbing states.}

\hide{\subsection{Extension to $n$ players} Assume $\a^*=(a_1^*,\dots,a^*_n)$ to be an optimal action profile. Now, define $\ubar{x}_k:=\min_{a_k \in \A_k}\{W(a_k,a^*_{-k})\}$, i.e., the minimum system objective value when player $k$ unilaterally deviates; observe that $\sum_{j=1}^n \ubar{x}_j \ge 1$ by Lemma \ref{lem_min_entry}. Further, let $\ubar{x}_{j:n}:= \sum_{i=j}^n \ubar{x}_i$.

Define $\mathcal{R}^*_\beta$ to be the class of action profiles that contains optimal action profile $\a^*$ and that is recurrent under some stationary distribution reached by Algorithm \ref{algo}. Then, we immediately have the following result 

\begin{theorem}\label{thm_rec_state_np}
   For any $\a \in \mathcal{R}^*_\beta$, the system objective function satisfies,
  \begin{eqnarray*}
      W(\a) &\ge& \max\{P_{n-1}-\beta, \ubar{x}_n-\beta\}\ \mbox{ where,}\\
      P_k&:=& \max\{P_{k-1}-\beta, \ubar{x}_{k}-\beta\} \ \mbox{ and},\\
      P_1&:=& 1-\ubar{x}_{2:n} - \beta.
  \end{eqnarray*} 
\end{theorem}

\begin{corollary}\label{cor_rec_state_np}
    The minimum value that can be achieved on $\mathcal{R}^*_\beta$ has the following lower bound $W(\a) \ge \frac{1}{n}-\frac{(n+1)\beta}{2}$ for all  $\a \in \mathcal{R}^*_\beta$.
\end{corollary}

{\color{red} -- not proved\begin{theorem}\label{thm_abs_state_np}
     For any class of action profiles $\mathcal{R}_\beta$ that is recurrent and does not contain optimal action profile,  we have $W(\a) > \frac{1}{2}+ \beta$ for all $\a \in \mathcal{R}_\beta$.
\end{theorem}}}

\section{Numerical Simulations} \label{sec: numerical}
As already mentioned, the sequence of action profiles $\{\a(t)\}$ generated using ABRA need not converge to a single absorbing state, and may  have multiple recurrent classes depending upon $(\beta,p)$. Thus, it is natural to  analyse the expected value of the system objective $W$ under achievable stationary distributions. Let $\Pi$ be the set of stationary distributions of the Markov chain $\{\a(t)\}_t$ concentrated on various recurrent classes and $E_\pi[\cdot]$ represent the expectation under distribution $\pi$. Define the following performance metric:
\begin{equation}\label{eqn_perf_metric}
  \P{W}{\beta,p}  := \frac{\min_{\pi \in \Pi}E_\pi[W]}{\max_{\a} W(\a)} = \min_{\pi \in \Pi}E_\pi[W],
\end{equation}as $\max_{\a} W(\a)=1$ in our framework. This metric measures the worst case expected system objective as compared to the optimal value.

\subsection{Expected system objective and Rationality parameter}\label{sec:num_rationality_parameter}
In this section, we numerically show the variation in probability of choosing an optimal action under stationary distribution concentrated over optimal class, $q^*_p$, as the rationality parameter varies. We also plot the expected system objective as a function of rationality parameter. Consider the following system objective function, 
\begin{equation}\label{eqn_W_ABRA_avoid_optimal_class_NE}
    W=\begin{bmatrix}
        1  &   0.8 &   0.6  &  0.5  \\
    0.5 &  0.6  &   0.4  &  0.39 \\
    0.5  &  0.39 &   0.39  &  0.39\\
    0.5 &  0.39  &  0.39   & 0.71
    \end{bmatrix},
\end{equation}which is controlled by actions of two players with action sets, $\A_i=\{a_i^1,a_i^2,a_i^3,a_i^4\}$ for $i =1,2$.  The $(j,k)$-th entry of matrix in \eqref{eqn_W_ABRA_avoid_optimal_class_NE} corresponds to the system objective value achieved under action profile $(a_1^j,a_2^k)$, for example $W(a^1_1,a^1_2)=1$ is the system optimal. Recall that the system objective function is the potential function for the corresponding potential game induced by marginal utilities. It is well known that any NE of the potential game is the local maximizer of the potential function (see \cite{shapley}). Thus, the game induced by  $W$ in \eqref{eqn_W_ABRA_avoid_bad_NE} has two NEs, $(a^1_1,a^1_2)$ with objective value $1$ and $(a_1^4,a_2^4)$ with objective value $0.71$; the noise parameter is fixed to $\beta=0.2$.

As the sequence of action profiles $\{\a(t)\}_t$ generated using ABRA  for $(\beta,p)$ progresses, it may converge to the absorbing state $(a_1^4,a_2^4)$ yielding a system objective value $0.71$ which is strictly more than $\frac{1}{2}+\beta$ as shown in Theorem \ref{thm_NE}. Otherwise, $\{\a(t)\}_t$ may  converges to the optimal class, which for $\beta=0.2$ contains the action profiles with following values, $\{1,0.8,0.6,0.5,0.4\}$. The minimum value that the system objective can take is 0.4, which is less than  $\frac{1}{2}$, and is worse than well known PoA bound in \cite{vetta}.
\begin{figure}
    \centering
    \includegraphics[trim ={ 3.5cm 8cm 3cm 9cm}, clip, scale=0.45]{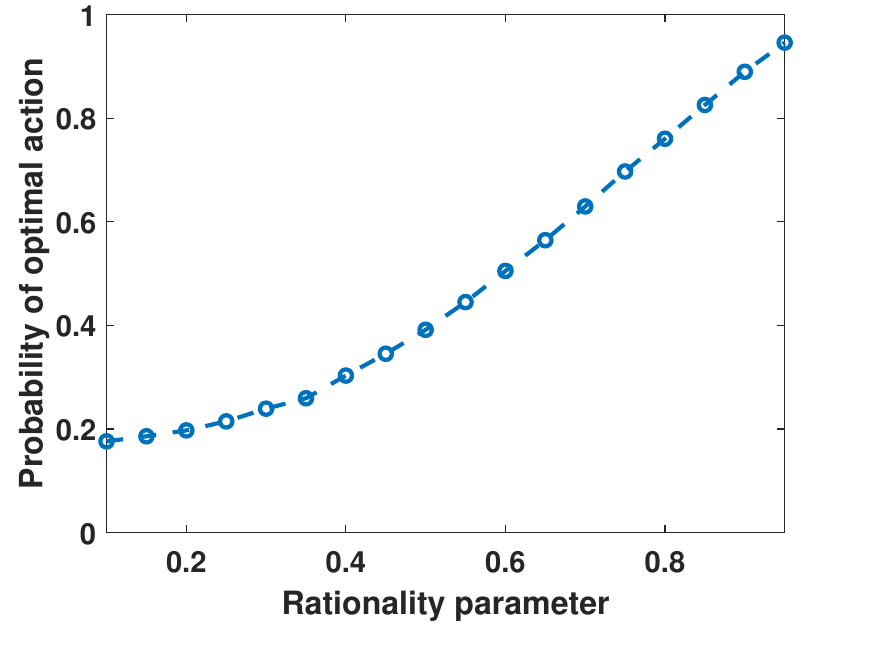}
    \caption{The probability of choosing optimal action under the stationary distribution concentrated on optimal class for $W$ in \eqref{eqn_W_ABRA_avoid_optimal_class_NE} of section \ref{sec:num_rationality_parameter}, when  noise parameter $\beta=0.2$. Thus, time spent in optimal action profile is more compared to other action profiles as rationality parameter $p$ increases.}
    \label{fig:stationary_dist}\hide{
\vspace{5mm}

    \includegraphics[trim ={ 3.5cm 8cm 3cm 9cm}, clip, scale=0.5]{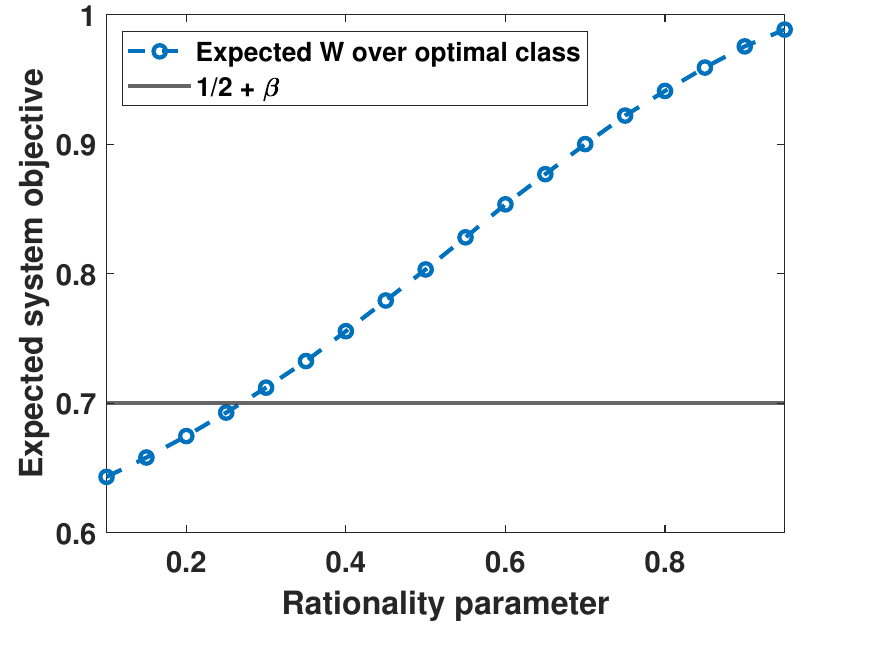}
    \caption{Expected value of system objective under stationary distribution concentrated on optimal class for $W$ in \eqref{eqn_W_ABRA_avoid_optimal_class_NE} of section \ref{sec:num_rationality_parameter}, when  noise parameter $\beta=0.2$. The expected value $E_{\pi^*_p}[W]$  increases with the rationality parameter $p$.}
    \label{fig:expected_W}}
\vspace{5mm}

     \includegraphics[trim ={ 3.5cm 8cm 3cm 9cm}, clip, scale=0.45]{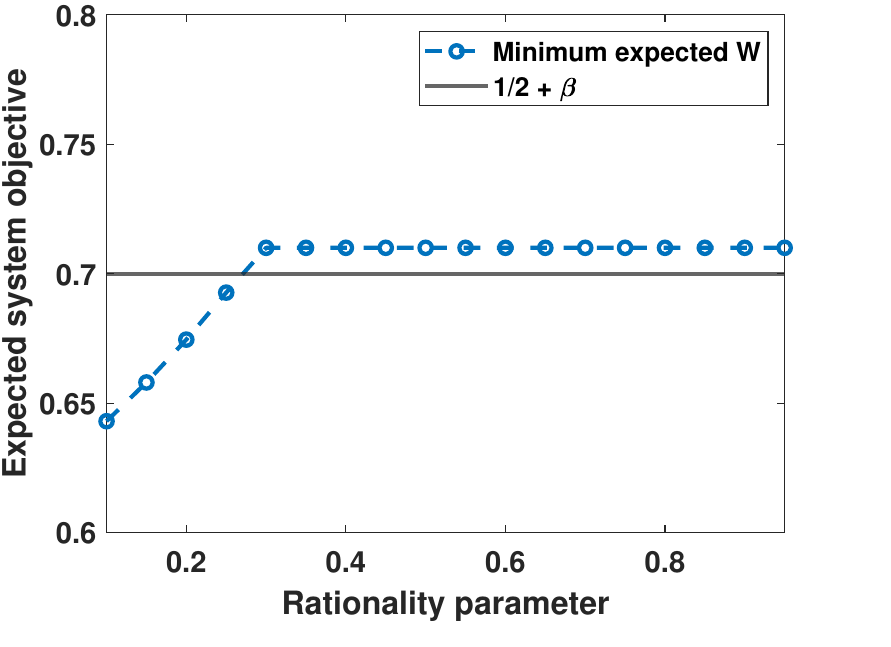}
     \caption{Minimum expected value of system objective $W$ in \eqref{eqn_W_ABRA_avoid_optimal_class_NE} of section \ref{sec:num_rationality_parameter}, when  noise parameter $\beta=0.2$ as a function of rationality parameter $p$. The minimum expected value is within $\frac{1}{2}+\beta$ of the optimal once the rationality parameter is sufficiently high.}
     \label{fig:min_exp_W}
 \end{figure}
However, in Figure \ref{fig:stationary_dist}, we observe that the probability of choosing optimal action  (when $\{\a(t)\}_t$ converges to optimal class) increases as the rationality parameter $p$ increases.  This implies that ABRA spends most of the time in optimal action profile whenever converged to optimal class.  \hide{Figure \ref{fig:expected_W} shows that the expected value of system objective under stationary distribution $\pi^*_p$ concentrated on optimal class, $E_{\pi^*_p}[W]$, also increases with the rationality parameter. Further, $E_{\pi^*_p}[W]> \frac{1}{2}+ \beta$ when rationality parameter $p \ge 0.3$; this is due to  the fact that probability of choosing optimal action, $q^*_p$, increases with $p$. One may anticipate that the minimum expected value of the system objective $ \P{W}{\beta,p} $ defined in see \eqref{eqn_perf_metric} should also improve with the rationality parameter.}  In Figure \ref{fig:min_exp_W}, we plot $ \P{W}{\beta,p} $ as a function of rationality parameter $p<1$ and observe that the minimum expected value of the system objective $ \P{W}{\beta,p} $ in  \eqref{eqn_perf_metric} improves with the rationality parameter $p$ and is above $\frac{1}{2}+\beta$ once $p\ge 0.3$.

 \subsection{Avoiding the bad NE using ABRA}\label{sec:escaping_bad_ne}
Now, we present an example where ABRA avoids a bad NE. We consider the following system objective function,
\begin{equation}\label{eqn_W_ABRA_avoid_bad_NE}
    W=\begin{bmatrix}
        1 & 0.8 & 0.5\\
         0.5   &0.6  &   0.5\\
    0.5 &  0.5  &  0.51
    \end{bmatrix},
\end{equation}which is controlled by actions of two players with action sets, $\A_i=\{a_i^1,a_i^2,a_i^3\}$ for $i =1,2$.  Here $W(a^1_1,a^1_2)=1$ is the system optimal. The game induced by  $W$ in \eqref{eqn_W_ABRA_avoid_bad_NE} has two NEs, the good NE $(a^1_1,a^1_2)$ with objective value $1$ and  the bad NE $(a_1^3,a_2^3)$ with objective value $0.51$.

 \begin{figure}[h]
    \centering
    \includegraphics[trim ={ 3.5cm 8cm 3cm 9cm}, clip, scale=0.45]{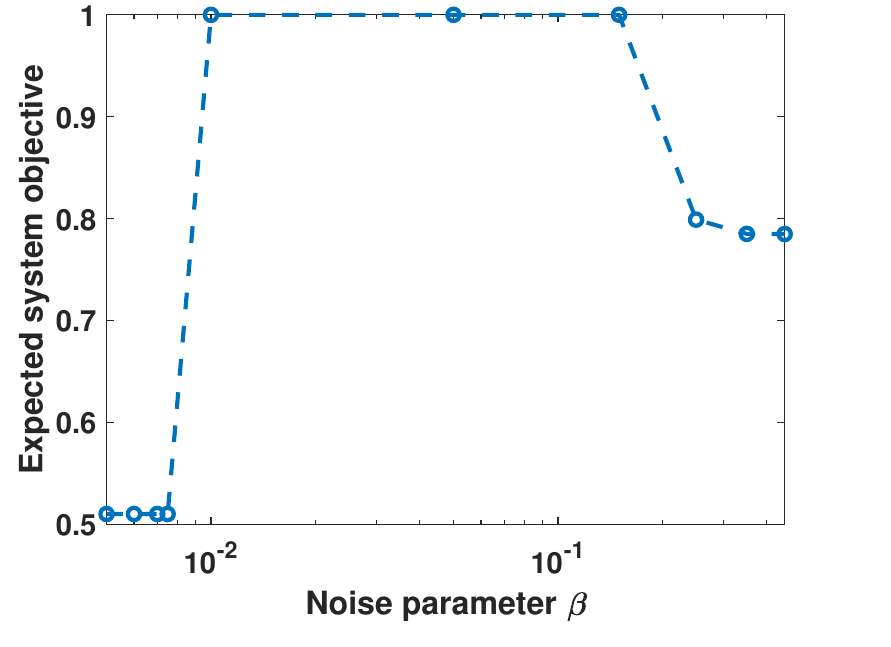}
    \caption{Minimum expected value of system objective $W$  in \eqref{eqn_W_ABRA_avoid_bad_NE} of section \ref{sec:escaping_bad_ne} as a function of noise parameter $\beta$ -- addition of small noise escapes the worst NE with $W=0.51$, large value of noise parameter bring down the system performance}
    \label{fig:beta_avoid_bad_NE}
\end{figure}

 In Figure \ref{fig:beta_avoid_bad_NE}, we plot the minimum expected value of system objective $W$, i.e., $\P{\beta,p}{W}$ of \eqref{eqn_perf_metric} 
 as a function of $\beta$ and fix the rationality parameter at $p=0.5$.   When the noise parameter $\beta$ is very small, it is not sufficient to escape the bad NE, and the agents may converge to  the  $(a_1^3,a_2^3)$, which yields minimum possible expected system objective $0.51$. When noise parameter is slightly  increased, e.g., $\beta =0.1$, then the NE $(a^3_1,a^3_2)$ becomes   transient, and the optimal action profile becomes the only absorbing state guaranteeing the optimal payoff of 1. However,  When the noise is further increased, the quality of the optimal class degrades, and all the action profiles constitute one optimal class: the addition of sub-optimal action profiles in the recurrent class brings $\P{\beta,p}{W}$ down. Interestingly, even the large noise case outperforms the zero noise case for an appropriate choice of the rationality parameter.

\section{Conclusions}
We propose an Approximate Best Response Algorithm (ABRA) to solve sub-modular maximization problems using game theory.  The ABRA is governed by a noise parameter and a rationality parameter. The noise parameter allows for the exploration and helps escaping the local maxima. On the other hand, the rationality parameter allows for exploitation of best response actions and ensures that the addition of noise does not degrade the system objective significantly. We prove for two player games that the minimum expected value of the system objective if ABRA reaches  to an equilibrium is strictly better than $\frac{1}{2}$ of the optimal by a term that can be controlled by noise parameter, and improves upon known price of anarchy bounds for such problems. Further, when ABRA converges to a recurrent class, we provide the lower bound on the system objective in such classes. Using numerical simulations, we show that the minimum expected value of the system objective under ABRA is more than half of the optimal by a term controlled by the noise parameter and the rationality parameter.   In future work, we aim to extend the findings of this work to a general $n$-player game.



\section*{APPENDIX}
{\begin{lemma}\label{lem_min_entry}
    If $W$ is sub-modular, non-decreasing and normalised, i.e. satisfies \eqref{eqn_sub_mod}, then for any $\a=(a_1,a_2)$,
    \vspace{-2mm}
    \begin{equation*}
        W(\a)  \le  W(a_1,\tilde{a}_2) +W(\tilde{a}_1,a_2) \mbox{ for all } \tilde{a}_1 \in \A_1,\ \tilde{a}_2 \in \A_2.
    \end{equation*}
\end{lemma}
\noindent\textbf{Proof:} From \eqref{eqn_sub_mod},  we have
\begin{eqnarray*}
    W(a_1,a_2) - W(a_2)\hspace{8mm} &\hspace{-20mm}\le&\hspace{-12mm} W(\varnothing,a_{1}) - W(\varnothing) = W(a_1), \\
    W(\a)  \le W(a_1)+W(a_2) &\le & W(a_1,\tilde{a}_2) +W(\tilde{a}_1,a_2).
\end{eqnarray*} In the above, the first inequality follows since $W$ is sub-modular. The second inequality follows since $W(\varnothing)=0$, and the last inequality follows from the non-decreasing nature of $W$. Hence the proof. \eop}

\hide{\begin{lemma}\label{lem_min_entry}
    If $W$ is sub-modular, non-decreasing and normalised, i.e. satisfies \eqref{eqn_sub_mod}, then for any $\a=(a_1,\dots,a_n)$,
    $$W(\a)  \le  W({a}_i,\tilde{a}_{-i}) +W(\tilde{a}_i,a_{-i}), $$
    for any $i$, $\tilde{a}_i \in \A_i$ and $\tilde{a}_{-i}$.
\end{lemma}
\noindent\textbf{Proof:} Fix $i$. From \eqref{eqn_sub_mod},  we have for any $j$,
\begin{eqnarray*}
    W(a_i,a_j,a_{-\{i,j\}}) - W(a_j,a_{-\{i,j\}}) &\le& W(\varnothing,a_{i}) - W(\varnothing), \\
    W(\a) - W(a_j,a_{-\{i,j\}}) & \le &W(a_i),\\
    W(\a) - W(\tilde{a}_i,a_{-i}) & \le &W(a_i,\tilde{a}_{-i}).
\end{eqnarray*} In the above, the first inequality follows since $W$ is sub-modular. The second inequality follows since $W(\varnothing)=0$, and last inequality follows from the non-decreasing nature of $W$. Hence the proof. \eop
}
\noindent\textbf{Proof of Theorem \ref{thm_NE}:} For any $j \in \A_2$, first observe by the definition of $\ubar{x}$ in \eqref{eqn_x_bar} and Lemma \ref{lem_min_entry} that
\begin{eqnarray}\label{eqn_low_min_entry}
     W(a_1^1,a_2^j) &&\hspace{-8mm} +\ \ubar{x} \ \ge\  W(a_1^1,a_2^1) =1, \mbox{ thus}\nonumber\\
   W(a_1^1,a_2^j) &\ge&  1 -\ubar{x}, \mbox{ and again by \eqref{eqn_x_bar}},\nonumber\\ 
   W(a_1^1,a_2^j)&\ge& \max\{\ubar{x},1-\ubar{x}\}\ \ge\  \frac{1}{2} \ \forall \ j \in \A_2,
\end{eqnarray}because $\max\{\ubar{x},1-\ubar{x}\}\ge  \frac{1}{2}$ for all $ \ubar{x} \in [0,1]$.
If $(a_1^*,a_2^*)$ is an absorbing state for the Markov chain generated using ABRA with noise parameter $\beta$, then it is trivially an NE, and $
    W(a_1^*,a_2^*) - \beta > W(a^1_1,a_2^*) \ge \frac{1}{2} $ by \eqref{eqn_low_min_entry}.  \eop

\noindent\textbf{Proof of Theorem \ref{thm_abs_state}:} Let $(a_1,a_2) \in \R_\beta$. Then\footnote{If any of \eqref{eqn_2p_12} or \eqref{eqn_2p_13} do not hold, then $(a_1,a_2)$ can lead to the optimal action profile $(a_1^1,a_2^1)$, violating the definition of the sub-optimal class. If~\eqref{eqn_2p_6} and \eqref{eqn_2p_7} both do not hold then $(a_1,a_2)$ cannot be in $\R_\beta$.}
\begin{eqnarray}
    W(a_1^1,a_2) &<& \max_{a \in \A_1} W(a, a_2) -\beta,\label{eqn_2p_12} \\
  \mbox{and} \ \  W(a_1,a_2^1) &<& \max_{b \in \A_2} W(a_1, b) -\beta, \label{eqn_2p_13}
\end{eqnarray}and at least  one of the following must hold:
\begin{eqnarray}
    W(a_1,a_2) &\ge& \max_{a \in \A_1} W(a, a_2) -\beta,\label{eqn_2p_6} \\
  \mbox{and/or} \ \  W(a_1,a_2) &\ge& \max_{b \in \A_2} W(a_1, b) -\beta. \label{eqn_2p_7}
\end{eqnarray}Say \eqref{eqn_2p_6} holds. Then by \eqref{eqn_low_min_entry} and \eqref{eqn_2p_12},
\begin{equation}\label{eqn_15_holds}
  \frac{1}{2}\ \le \ W(a_1^1,a_2) < \max_{a \in \A_1} W(a, a_2) - \beta \ \le \  W(a_1,a_2).
\end{equation}Thus, we only need to consider the case when  \eqref{eqn_2p_7} holds and \eqref{eqn_2p_6} does not -- we must have some $\tilde{b} \in \A_2$ with $(a_1,\tilde{b}) \in \R_\beta$ such that\footnote{When \eqref{eqn_2p_6} does not hold, $(a_1,a_2)$ is not reachable by action of player 1, thus it much be reachable by action of player 2 -- thus, there must be some action profile $(a_1,\tilde{b})\in \R_\beta$ reachable by action of player 1 such that player 2 can reach $(a_1,a_2)$ from $(a_1,\tilde{b})$.},
\begin{eqnarray}\label{eqn_2p_11}
   W(a_1,\tilde{b}) &\ge &  \max_{a \in \A_1} W(a,\tilde{b}) - \beta \ > \ W(a^1_1,\tilde{b}).
   \end{eqnarray}%
We now prove \eqref{eqn_bdd_min_rec} using contradiction. Say  $W(a_1,a_2)= \frac{1}{2}-\frac{\beta}{2}-\epsilon$ for some $\epsilon\ge 0$. 
Using equations \eqref{eqn_2p_13} and \eqref{eqn_2p_7}, 
\begin{eqnarray}
W(a_1,a^1_2) &<& \max_{b \in \A_2} W(a_1,b) - \beta \ \le \ \frac{1}{2} - \frac{\beta}{2} - \epsilon. \hspace{3mm}\label{eqn_2p_9}
\end{eqnarray}By Lemma \ref{lem_min_entry}, $W(a_1^1,\tilde{b}) \ge 1 - W(a_1,a_2^1) > \frac{1}{2} + \frac{\beta}{2} + \epsilon  $. Further by \eqref{eqn_2p_11} and \eqref{eqn_2p_9}, we have the contradiction,
\begin{eqnarray*}
 \frac{1}{2} + \frac{\beta}{2} + \epsilon  &<& W(a_1^1,\tilde{b}) \ <\  W(a_1,\tilde{b}) \ \le\ \frac{1}{2} + \frac{\beta}{2} -\epsilon.
\end{eqnarray*}

Towards proving \eqref{eqn_new_max_two_p}, say  $W(a_1,a_2) = \frac{1}{2} - \delta$ for $0 \le \delta < \frac{\beta}{2}$. \hide{Let  $(a_1,a_2) \in \R_\beta$ be such that $W(a_1,a_2) = \frac{1}{2} - \delta$ for $0 \le \delta < \frac{\beta}{2}$. From the definition of $\R_\beta$, at least one of the \eqref{eqn_2p_6} or \eqref{eqn_2p_7} must hold. First consider the case when \eqref{eqn_2p_6} holds. This implies, further using  \eqref{eqn_low_min_entry} and \eqref{eqn_2p_12},
\begin{eqnarray*}
  \frac{1}{2}\ \le \ W(a_1^1,a_2) &<& \max_{a \in \A_1} W(a, a_2) - \beta \ \le \ \frac{1}{2}  - \delta,
\end{eqnarray*}which provides a contradiction. Thus, $\min_{a \in \R_\beta} W(a) > \frac{1}{2}$ whenever \eqref{eqn_2p_6} holds. Now consider the only remaining case, when \eqref{eqn_2p_7} holds and \eqref{eqn_2p_6} does not.}By \eqref{eqn_15_holds}, equation \eqref{eqn_2p_6}  must not hold, thus \eqref{eqn_2p_7} should be true. By \eqref{eqn_2p_13}  and  \eqref{eqn_2p_7}, 
\begin{eqnarray}
   \frac{1}{2}-\delta = W(a_1,a_2) \ \ge\  \max_{b \in \A_2} W(a_1, b) -\beta &>&  W(a_1,a^1_2).\nonumber
\end{eqnarray}Thus, by Lemma \ref{lem_min_entry}, $ W(a^1_1,a_2) > \frac{1}{2} +\delta$.
\hide{\begin{eqnarray}\label{eqn_2p_3}
     W(a^1_1,a_2) &>& \frac{1}{2} +\delta.
\end{eqnarray}}Using \eqref{eqn_2p_12}, 
\begin{eqnarray*}
  \frac{1}{2} +\delta \ <\  W(a_1^1,a_2) &<& \max_{a \in \A_1} W(a, a_2) -\beta \\
  \max_{a \in \A_1} W(a, a_2)&>&\frac{1}{2} + \delta + \beta,
\end{eqnarray*}and since $\arg\max_{a \in \A_1} W(a, a_2) \in \R_\beta$ whenever $(a_1,a_2) \in \R_\beta$, \eqref{eqn_new_max_two_p} is true.   \eop

\noindent\textbf{Proof of Lemma \ref{thm_rec_state}: } Fix $\beta$ and consider the problems with non-empty $\mathcal{R}_\beta^*$. Consider any action profile $(a_1,a_2)$ in $\mathcal{R}^*_\beta$. If $a_1 = a^1_1$, then by \eqref{eqn_low_min_entry}, 
\begin{eqnarray*}
    W(a_1,a_2) &\ge& \max\{\ubar{x},1-\ubar{x}\},\\
    &\ge& \max\{\max\{1-\ubar{x}, \ubar{x}\}-2\beta,\ubar{x}-\beta\}.
\end{eqnarray*}
Now, say $a_1 \ne a^1_1$. Observe that for all $\tilde{a}_2 \in \A_2$, the set of action choices for player 1, $\B_1(\tilde{a}_2)\cup\nbd{1}{\tilde{a}_2 }$ always guarantees a payoff more than $\max\{1-\ubar{x},\ubar{x}\}-\beta$ using \eqref{eqn_low_min_entry}. This implies that $(a_1,a_2) \in \mathcal{R}_\beta^*$ under ABRA only if $$W(a_1,a_2')\ge \max\{1-\ubar{x},\ubar{x}\}-\beta \mbox{ for some } a_2' \in \A_2.$$%
Then, the $\beta$-neighbourhood of best responses for player 2, $\B_2(a_1)\cup \nbd{2}{a_1 }$ guarantees a payoff more than $\max\{\max\{1-\ubar{x},\ubar{x}\}-2\beta,\ubar{x}-\beta\}$.  Since the agents can never move to an action profile worse than the ones in $\beta$-neighbourhood under ABRA, and since the above is true for any $(a_1,a_2)$, we have the proof. \eop
\hide{\vspace{3cm}
First consider the case when $\max\{1-\ubar{x}, \ubar{x}\}-2\beta> \ubar{x}-\beta$, which implies $1-\ubar{x}> \ubar{x}$. 

Say $(a_1,a_2)$ is an action profile such that $W(a_1,a_2) < 1- \ubar{x}-2\beta$. If $(a_1,a_2)$ is in $\mathcal{R}_b$ then . Then, if player 1 moves, the next action profile $(a',b)$is guaranteed to have the payoff of at least $1-\ubar{x}-\beta$, due to lemma 1, which implies $W(a_1^1,b) \ge 1-\ubar{x}$ for any $b$. Thus, player 1 can never bring the system objective value lower than $1-\ubar{x}-\beta$. Further, when player 2 moves, even in the worst case scenario, it can not bring the system objective value to be less than $\max\{\max\{1-\ubar{x}, \ubar{x}\}-2\beta, \ubar{x}-\beta\}$ by the definition of $\ubar{x}$.

On the other hand, when $\max\{1-\ubar{x}, \ubar{x}\}-2\beta \le \ubar{x}-\beta$. Say there is an action profile $(a,b) \in \mathcal{R}^*_\beta$ such that $W(a,b) < \ubar{x}-\beta$. Then, if player 1 moves, the next action profile $(a',b)$is guaranteed to have the payoff of at least $\max\{1-\ubar{x}, \ubar{x}\}-\beta$, due to lemma 1 and the definition of $\ubar{x}$, which implies $W(a_1^1,b) \ge \max\{1-\ubar{x}, \ubar{x}\}-\beta \ge \ubar{x}-\beta$ for any $b$. Thus, player 1 can never bring the system objective value lower than $\ubar{x}-\beta$. Further, when player 2 moves, even in the worst case scenario, it can not bring the system objective value to be less than $\ubar{x}-\beta\}$ by the definition of $\ubar{x}$. Hence the proof.  \eop}

\noindent\textbf{Proof of Theorem \ref{cor_rec_state}:} From Lemma \ref{thm_rec_state}, the minimum value that $W$ can take over $\mathcal{R}_\beta^*$ is $\max\{\max\{1-\ubar{x}, \ubar{x}\}-2\beta, \ubar{x} - \beta\}$. Hence, for $\ubar{x} \ge \frac{1}{2}$,  $W \ge \ubar{x}-\beta \ge \frac{1}{2}-\beta$. For $\ubar{x}<\frac{1}{2}$, $W \ge \max\{1-\ubar{x} - 2\beta, \ubar{x}-\beta\}$. The right hand side is minimized when $\ubar{x}= \frac{1}{2}-\frac{\beta}{2}$, and provides a minimum value of $\frac{1}{2}-\frac{3}{2}\beta$. Thus, the result. \eop

{\noindent\textbf{Proof of Proposition \ref{prop_low_bound_subopt}:} The lower bound in Theorem \ref{cor_rec_state} is achieved when $\ubar{x}=\frac{1}{2}-\frac{\beta}{2}$. By Lemma \ref{lem_min_entry}, we get,
\begin{eqnarray}\label{eqn_pp_1}
     W(a^1_1,b) &>& \frac{1}{2} +\frac{\beta}{2} \mbox{ for all } b \in \A_2.
\end{eqnarray}Let $(a_1^*,a_2^*) \in \arg\max_{(a_1,a_2)\in\R_\beta} W(a_1,a_2)$. By \eqref{eqn_2p_12},
\begin{eqnarray*}
  \frac{1}{2} + \frac{\beta}{2} &<& \max_{a \in \A_1} W(a, a^*_2) -\beta\ = \  W(a_1^*,a_2^*) -\beta,
\end{eqnarray*}thus the proof.\eop}

\hide{From  \eqref{eqn_low_min_entry}, we have $W(a_1^1,a_2) > 1-\ubar{x} = \frac{1}{2}+ \frac{\beta}{2}$ for all $(a,b) \in \mathcal{R}_\beta$. Further, we also have that,
\begin{eqnarray*}
    \max_{(a,b)\in \mathcal{R}_\beta} W(a,b) -\beta &>& 1-\ubar{x} = \frac{1}{2}+\frac{\beta}{2},\\
    \mbox{thus, }  \max_{(a,b)\in \mathcal{R}_\beta} W(a,b) &>& \frac{1}{2}+\frac{3}{2}\beta.
\end{eqnarray*}Thus the proof. \eop}

\hide{noindent\textbf{Proof of Proposition \ref{prop_lower_bound_prob}:} 
From the definition of $q_p^*$ and Theorem \ref{cor_rec_state}, we have
\begin{eqnarray*}
    E_{\pi^*_p} [W] &\ge& q_p^* + (1-q_p^*) \left(\frac{1}{2}- \frac{3\beta}{2}\right),\\
    & =  &\frac{1}{2}+\beta  + q_p^* \left(\frac{1}{2} + \frac{3\beta}{2}\right)  -\frac{5\beta}{2},\\
    &> & \frac{1}{2}+\beta\ \   \mbox{ when } q_p^*\  > \ \frac{5\beta}{1+3\beta}.
\end{eqnarray*} \eop}




\begin{thebibliography}{99}
\bibitem{Qu}
Qu, G., Brown, D., \& Li, N. (2019). Distributed greedy algorithm for multi-agent task assignment problem with submodular utility functions. Automatica, 105, 206-215.

\bibitem{eitan}
Altman, E., El Azouzi, R., \& Jiménez, T. (2004). Slotted Aloha as a game with partial information. Computer networks, 45(6), 701-713.

\bibitem{liu}
Liu, C., Lu, K., Chen, X., \& Szolnoki, A. (2023). Game-theoretical approach for task allocation problems with constraints. Applied Mathematics and Computation, 458, 128251.

\bibitem{zhu}
Zhu, M., \& Martínez, S. (2013). Distributed coverage games for energy-aware mobile sensor networks. SIAM Journal on Control and Optimization, 51(1), 1-27.

\bibitem{Brown}
Brown, P. N., Seaton, J. H., \& Marden, J. R. (2023). Robust networked multiagent optimization: designing agents to repair their own utility functions. Dynamic Games and Applications, 13(1), 187-207.
\bibitem{ferguson}
Ferguson, B. L., Brown, P. N., \& Marden, J. R. (2020, July). Carrots or sticks? the effectiveness of subsidies and tolls in congestion games. In 2020 American Control Conference (ACC) (pp. 1853-1858). IEEE.

\bibitem{Kordonis}
Kordonis, I., Dessouky, M. M., \& Ioannou, P. A. (2019). Mechanisms for cooperative freight routing: Incentivizing individual participation. IEEE Transactions on Intelligent Transportation Systems, 21(5), 2155-2166.

\bibitem{marden}
Marden, J. R., \& Shamma, J. S. (2015). Game theory and distributed control. In Handbook of game theory with economic applications (Vol. 4, pp. 861-899). Elsevier.

\bibitem{jaleel}
Jaleel, H., \& Shamma, J. S. (2020). Distributed optimization for robot networks: From real-time convex optimization to game-theoretic self-organization. Proceedings of the IEEE, 108(11), 1953-1967.

\bibitem{martin}
Martin, J. G., Muros, F. J., Maestre, J. M., \& Camacho, E. F. (2023). Multi-robot task allocation clustering based on game theory. Robotics and Autonomous Systems, 161, 104314.

\bibitem{marden_util_des}
Paccagnan, D., Chandan, R., \& Marden, J. R. (2019). Utility design for distributed resource allocation—part i: Characterizing and optimizing the exact price of anarchy. IEEE Transactions on Automatic Control, 65(11), 4616-4631.

\bibitem{marden_potential}
Leslie, D. S., \& Marden, J. R. (2011, October). Equilibrium selection in potential games with noisy rewards. In International Conference on NETwork Games, Control and Optimization (NetGCooP 2011) (pp. 1-4). IEEE.

\bibitem{PoA}
Papadimitriou, C. (2001, July). Algorithms, games, and the internet. In Proceedings of the thirty-third annual ACM symposium on Theory of computing (pp. 749-753).

\bibitem{vetta}
Vetta, A. (2002, November). Nash equilibria in competitive societies, with applications to facility location, traffic routing and auctions. In The 43rd Annual IEEE Symposium on Foundations of Computer Science, 2002. Proceedings. (pp. 416-425). IEEE.

\bibitem{seaton}
Seaton, J. H., \& Brown, P. N. (2023). On the Intrinsic Fragility of the Price of Anarchy. IEEE Control Systems Letters.

\bibitem{marden_logit}
Marden, J. R., \& Shamma, J. S. (2012). Revisiting log-linear learning: Asynchrony, completeness and payoff-based implementation. Games and Economic Behavior, 75(2), 788-808.

\bibitem{tatarenko}
Tatarenko, T. (2014, December). Log-linear learning: Convergence in discrete and continuous strategy potential games. In 53rd IEEE Conference on Decision and Control (pp. 426-432). IEEE.


\bibitem{Levin}
Levin, D. A., \& Peres, Y. (2017). Markov chains and mixing times (Vol. 107). American Mathematical Soc.

\bibitem{Norris}
Norris, J. R. (1998). Markov chains (No. 2). Cambridge university press.

\bibitem{shapley}Monderer, D., \& Shapley, L. S. (1996). Potential games. Games and economic behavior, 14(1), 124-143.
\end{thebibliography}
\end{document}